\documentclass[12pt]{article}
\usepackage{amsfonts}
\usepackage{amsmath, amsthm, upref}
\usepackage{enumerate}
\usepackage[parfill]{parskip}
\usepackage{color}
\usepackage{hyperref}
\usepackage{multimedia}
\usepackage{graphics}
\usepackage{epsfig}
\usepackage{graphicx}
\usepackage[greek,english]{babel}
\usepackage{textcomp}
\usepackage{eurosym}
\usepackage{enumerate}

\DeclareMathOperator{\e}{e}

\newcommand{\mbf}{\mathbb{F}}

\newcommand{\comment}[1]{}

\newcommand{\mcf}{\mathcal{F}}

\newtheorem{theorem}{Theorem}

\newtheorem{corollary}{Corollary}

\newtheorem{example}{Example}

\newtheorem{proposition}{Proposition}
\newtheorem{remark}[theorem]{Remark}

\newcommand{\bee}{\begin{equation}}
\newcommand{\eee}{\end{equation}}
\newcommand{\bea}{\begin{eqnarray}}
\newcommand{\eea}{\end{eqnarray}}
\newcommand{\bean}{\begin{eqnarray*}}
\newcommand{\eean}{\end{eqnarray*}}

\newcommand{\N}{\mathbb{N}}

\newcommand{\Z}{\mathbb{Z}}

\begin{document}

\title{Strict Local Martingales and the Khasminskii Test for Explosions}
%\author{Philip Protter%
%\thanks{Statistics Department, Columbia University, New York, NY 10027; email: pep2117@columbia.edu.%
%}  and Aditi Dandapani%
%\thanks{ Applied Mathematics Department, Columbia University, New York, NY 10027; email: ad2259@caa.columbia.edu
%}}
\author{Aditi Dandapani%
\thanks{Applied Mathematics Department, Columbia University, New York, NY 10027; email: ad2259@caa.columbia.edu; Currently at Ecole Polytechnique, Palaiseau, France.}
\thanks{Supported in part by NSF grant DMS-1308483
} ~and Philip Protter%
\thanks{Statistics Department, Columbia University, New York, NY 10027; email: pep2117@columbia.edu.
}
\thanks{Supported in part by NSF grant DMS-1612758}
}
\date{\today }
\maketitle

\begin{abstract}
We exhibit sufficient conditions such that components of a multidimensional SDE giving rise to a local martingale $M$ are strict local martingales or martingales. We assume that the equations have diffusion coefficients of the form $\sigma(M_t,v_t),$ with $v_t$ being a stochastic volatility term. 

\end{abstract}

\subsection*{We dedicate this paper to the memory of Larry Shepp}

\section{Introduction}\label{intro}

In 2007 P.L. Lions and M. Musiela~(2007) gave a sufficient condition for when a unique weak solution of a one dimensional stochastic volatility stochastic differential equation is a strict local martingale.  Lions and Musiela also gave a sufficient condition for it to be a martingale. Similar results were obtained by L. Andersen and V. Piterbarg~(2007). The techniques depend at one key stage on Feller's test for explosions, and as such are uniquely applicable to the one dimensional case. Earlier, in 2002, F. Delbaen and H. Shirakawa~(2002) gave a seminal result giving necessary and sufficient conditions for a solution of a one dimensional stochastic differential equation (without stochastic volatility) to be a strict local martingale or not. This was later extended by A. Mijatovic and M. Urusov~(2012) and then by Cui et. al.~(2017) to the case where there exists a unique, weak solution, a well as stochastic volatility. 

Our idea is to use the theory of Lyapunov functions along the lines of Khasminskii (1980), Narita (1982), and Stroock-Varadhan (2006). This builds on the previous seminal work of J. Ruf (2015). We recently became aware of the excellent paper of D. Criens (2018) which uses techniques that are similar to those used in this paper. Criens is more interested in conditions for the absence of arbitrage, whereas we are more interested in the nature of martingales versus strict local martingales within a system of stochastic differential equations. We give a criterion, in terms of the coefficients of the equation, as to which components are martingales, and which are strict local martingales. Our criterion is deterministic. In addition, we consider a stochastic volatility case. 

In this paper we give conditions such that components of local martingale solutions to multidimensional stochastic differential equations are true martingales or not.  To incorporate the notion of stochastic volatility, we allow the diffusion coefficient of the local martingale in question to be a function of the multidimensional local martingale itself, as well as another stochastic process, namely, the stochastic volatility. In the models we study here, the concept of explosions plays a crucial role in our analysis, but due to the multidimensional nature of the setting we cannot any longer use Feller's test for explosions. This is because Feller's test for explosions can only tell us whether or not the solution to a one dimensional stochastic differential explodes, and we are dealing here with $n$ dimensional where $n>1$. The fact that the diffusion coefficient is a function of the vector local martingale itself, as well as of the stochastic volatility, brings us away from the realm of Feller's test for explosions, which is a solely one-dimensional phenomenon. 
Rather, we will rely on the use of the theory of multidimensional diffusions, as expounded in  Stroock-Varadhan (2006) and Narita~(1982). The theory of explosions of multidimensional diffusions involves Lyapunov functions and is a part of the study of stability theory, as established, for example, long ago by Khasminsikii~(2012). The link between explosions and the martingale property is crucial in the one-dimensional case, and we will see that it continues to be this way in the multidimensional case.

There is already a fairly keen interest in the topic of strict local martingales . Some relatively recent papers concerning the topic  include Biagini et al~(2014), Bilina-Protter~(2015), Chybiryakov~(2007), Delbaen-Schachermayer~(1998), X-M Li~(2017), Lions-Musiela~(2007), Hulley~(2010), Keller-Ressel~(2014), Madan-Yor~(2006), Mijatovic-Urusov~(2012), and Sin~(1998). An incentive for studying strict local martingales is their connection to the analysis of financial bubbles. Mathematical finance theory tells us that on a compact time set, the (nonnegative) price process of a risky asset is in a bubble if and only if the price process is a strict local martingale under the risk neutral measure governing the situation. 

There is only one other result that we know of that provides examples of multidimensional strict local martingales, and that is the very nice little paper of Xue-Mei Li~\cite{X-ML}. Her approach is quite different from ours. 

An outline of our paper is as follows. In the first section, we introduce multidimensional local martingales and the probabilistic setup. The link between explosion and the martingale property is established. In the next section, we lay out sufficient conditions for explosion or non-explosion of multidimensional processes.

We proceed after this to exhibit some examples in the two and three dimensional cases.

Our most important results are Corollary~\ref{col} and Theorem~\ref{main}.

\section{The Case of a Vector of Strict Local Martingales}\label{vect}

Suppose that we have a filtered probability space $(\Omega, \mathcal{F}_T, (\mathcal{F}_t)_{t\in [0,T]}, P).$ Suppose also that we have a $d$-dimensional vector of $(P, \mathcal{F})$ local martingales, $M$, as well as a stochastic volatility process, $v$. We will call the $d+1$ dimensional process $\begin{bmatrix} M\\ v \end{bmatrix}=X.$ We assume the $d$ dimensional local martingale $M$ and the process $v$ solve the following stochastic differential equations (we drop the vector notation for convenience):

 \begin{eqnarray} dM_t&=&M_t\sigma(M_t,v_t)dB_t; \quad M_0=1\label{d}\\
 dv_t&=&\bar{\sigma}(M_t,v_t)dZ_t +b(M_t,v_t)dt; \quad v_0=1\label{v} \end{eqnarray}   

In the above, we assume that $\sigma$ is a $d \times d$ diagonal matrix of diffusion coefficients that is locally bounded and measurable. $B$ is a $d$ dimensional correlated Brownian motion. We'll also assume that $\bar{\sigma}$ and $b$ are locally bounded and measurable.

% Note that the vector process $X=\begin{bmatrix} M\\ v \end{bmatrix}$ does not explode.
 Assume also that $W:=\begin{bmatrix} B \\Z \end{bmatrix}$ is a $d+1$ dimensional correlated Brownian motion with correlation matrix $\Sigma.$ Lastly, our time interval is $[0,T].$

We would like to answer the following question: what are sufficient conditions such that components $1:j$ of $M$ are true martingales and components $j+1:d$ are strict local martingales? Before we answer this question, let us note that, for each $j,$ $M_j$ is a non-negative local martingale and hence a supermartingale by Fatou's Lemma, and it will be a true martingale on the time interval $[0,T]$ if and only if $E[M^{j}_T]=M^{j}_0=1.$ 

We begin with the following setup:

Let $\Omega^{1}$ denote the space of continuous functions $\omega^{1} : \omega^{1}:  (0,\infty) \rightarrow [0, \infty] $ such that $0$ and $\infty$ are absorbing boundaries and $\omega^{1}(0)=1.$ 
 
 For $2 \leq j \leq d$ let $\Omega^{j}$ denote the space of continuous functions $\omega^{j}= (0,\infty) \rightarrow [0, \infty)$.  Define, for all $i \in \mathbb{N},$ $2\leq j\leq d+1$ $R^{j}_i=\inf \{t \in [0,T]: \omega^{j}_t \textgreater i\}$ and $S^{j}_i= \inf \{t \in [0,T] : \omega^{j}(t) \leq \frac{1}{i}\},$ with the convention that $\inf\{\phi \}=T.$ Then, let $T^{j}_\infty:= \lim_{ i \to \infty} R^{j}_i$ and $T^{j}_0:= \lim_{i \to \infty} S^{j}_i.$ $T^{j}_0$ and $T^{j}_\infty$ denote the hitting times of $0$ and $\infty$ of $\omega^{j}$ respectively.

Now we can better define  $\Omega^{j}$ to be the space of continuous functions $\omega^{j}= (0,\infty) \rightarrow [0, \infty)$ such that $\omega^{j}(0)=1$ and $\omega^{j}_t=\omega^{j}_{(t \wedge T^{j}_0 \wedge T^{j}_\infty)}.$  
 
 For $d+1\leq j\leq2d+1$ let $\Omega^{j}$ denote the space of continuous functions $\omega^{j}:[0, \infty) \rightarrow (-\infty, \infty)$ with $\omega^{j}(0)=0.$

Recalling that we called the $d+1$ dimensional Brownian motion $W=\begin{bmatrix} B \\Z \end{bmatrix},$ define the canonical process by $\begin{bmatrix} v \\ M\\ W \end{bmatrix}= \begin{bmatrix} \omega^{1}(t) \\ \omega^{j}(t)_{2 \leq j\leq d} \\ \omega^{j}(t)_{d+1\leq j \leq 2d+2}) \end{bmatrix} $ for all $t \geq 0.$ Let the filtration $ \mbf =\{\mathcal{F}_t\}_{t\in [0,T]}$ be the right-continuous filtration generated by the canonical process.

  Let $\Omega= \prod_{j=1}^{2d+2}\Omega_j$ and $\omega={\omega_j}_{1\leq j\leq 2d+2}.$ Henceforth, processes will be defined on the filtered space $(\Omega,\mathcal{F}_T,\mathcal{F}_t)$, $t \in [0, T].$ Let $P$ be the probability measure induced by the canonical process on the space $(\Omega,\mathbb{F}).$

 Given the canonical space $(\Omega, \mathcal{F}_T, (\mathcal{F}_t)_{t \in [0,T]}),$ the processes $(v, M,W)$ correspond to the $2d+2$ components of $\omega.$
 We assume that the processes $v$ and $M$ are adapted to the filtration $\{\mathcal{F}_t\}_{t \in [0,T]}$ as well as $W,$ which is a Brownian motion with respect to this same filtration.
 
Define by $N^{j}_t$ the continuous $(P,\mathbb{F})$ local martingale $N^{j}_t=\int_{0}^{t} \sigma_{j,j}(M_s, v_s)dB^{j}_s.$ Note that more generally $N^j$ is a sigma martingale, but since all continuous sigma martingales are local martingales, we dispense with the notion of sigma martingales in this paper. We have that the process $M^{j}$ is given by: $M^{j}=\mathcal{E}(N^{j}).$ Here by $\mathcal{E}(N^{j})$ we mean the stochastic exponential of $N^j$. (See for example~\cite[p. 85]{PP}.)
  
  For the convenience of the reader we recall the definition of a \textit{standard system} as defined in~\cite{PAR} and its implications: Let $\mathcal{T}$ be a partially ordered, non-void indexing set and let $(\mathcal{F}^{0}_t)_{t \in \mathcal{T}}$ be an increasing family of $\sigma$ fields on $\Omega.$ We say that $(\mathcal{F}^{0}_t)_{t \in \mathcal{T}}$ is a \textit{standard system} if: 
  
\begin{enumerate}
 \item Each measurable space $(\Omega, \mathcal{F}^{0}_t)$ is a standard Borel space. In other words, $\mathcal{F}^{0}_t$ is $\sigma$ isomorphic to the $\sigma$ field of Borel sets on some complete separable metric space. 
\item For any increasing sequence $t_i \in \mathcal{T},$ and decreasing sequence $\mathcal{A}_i \in \Omega$ such that $\mathcal{A}_i$ is an atom of $\mathcal{F}_{t_i},$ we have $\cap_{i} \mathcal{A}_i \neq \phi.$

\end{enumerate}

Let us also state the implications of the property of being a standard system: Let $\mathcal{F}_i$ be a sequence of $\sigma$ fields on $\Omega$ satisfying $(1)$ and let $\mu_i$ be a consistent sequence of probability measures on ${\mathcal{F}_i}_{(i \geq 1)}.$ Then, from Parthasarathy~(1967), we have the following theorem:

 \begin{theorem}\label{par}[Parthasarathy] If condition $(2)$ holds, then ${\mu_i}_{(i \geq 1)}$ admits an extension to $\bigvee_{i \geq 1}\mathcal{F}_i.$

\end{theorem}
  
For any $\mbf$ stopping time $\tau,$ we define $\mathcal{F}_{\tau_{-}}$ as the smallest $\sigma$ algebra containing $\mathcal{F}^{0}_0$ and all sets of the form $A\cap\{t<\tau\}$ for $t\in [0,T]$, and $A\in\mcf_t$. See for example~Protter~(2005), p. 105.

We have the following theorem, from~Carr et al~(2014).

\begin{theorem}\label{ext}
Consider the probability space $(\Omega, \mathcal{F}_T, (\mathcal{F}_t)_{t\in [0,T]}, P)$ with the process $M^{j}$ defined as in \eqref{d}, with $M^{j}_0=1.$ Then there exists a unique probability measure, call it $P^{j},$ on $(\Omega, \mathcal{F}_{T^{j}_{\infty}})$ such that, for any stopping time $0 \textless \nu \textless \infty,$

 \begin{enumerate}
  \item \begin{equation}\label{1.3} P^{j}(A \cap\{T^{j}_\infty \textgreater \nu \wedge T\})=E_{P}[1_{A}M^{j}_{\nu \wedge T}] \end{equation} for all $A \in \mathcal{F}_{\nu \wedge T}.$

\item For all non-negative $\mathcal{F}_{\nu \wedge T}$ measurable random variables $U$ taking values in $[0, \infty],$  \begin{equation}\label{1.4} E_{P^{j}}[U1_{\{{T^{j}_\infty} \textgreater \nu \wedge T\}}]=E_{P}[UM^{j}_{\nu \wedge T}1_{\{T^{j}_0 \textgreater \nu \wedge T\}}] \end{equation}

and with $\tilde{M}^{j}_t=\frac{1}{M^{j}_t}1_{\{T^{j}_\infty \textgreater t\}},$

\item \begin{equation}\label{1.5} E_{P}[U1_{\{T^{j}_0 \textgreater \nu \wedge T\}}]=E_{P^{j}}[{U\tilde{M}^{j}}_{\nu \wedge T}]  \end{equation}

\item $M^{j}$ is a true martingale if and only if \begin{equation}\label{explo}P^{j}(T^{j}_{\infty} \textgreater T) =1 \end{equation}

\end{enumerate}

\end{theorem}

\begin{proof}

Recall our assumption that $M^{j}_0=1.$ Observe that the stopped process $(M^{j})^{R^{j}_i}$ is a nonnegative martingale. 

Therefore, it generates a measure $P^{j}_i$ on $(\Omega, \mathcal{F}_{{T_\infty^{j}}_i})$ by $dP^{j}_i:=(M^{j})^{R^{j}_i}_TdP$ for all $i \in \mathbb{N}.$ Note that the family of probability measures $\{P^{j}_i\}$ is consistent for all $i,$ in that $\tilde{P}_{i+k}| \mathcal{F}_{{T_\infty^{j}}}=\tilde{P}_i$ $\forall i,k \in \mathbb{N}$ and $\mathcal{F}_{{T_\infty^{j}}}= \bigvee_{i \in \mathbb{N}} \mathcal{F}_{R^{j}_i}$

The extension theorem $V.4.1$ of~Parthasarathy~(1967) gives us the existence of a probability measure $P^{j}$ on $(\mathcal{F}_{{T_\infty^{j}}})$ such that $P^{j}|\mathcal{F}_{R^{j}_i}= P^{j}_i.$ Let us now check that the conditions of this theorem are indeed satisfied in our case. 

We need to check that $\{\mathcal{F}_{R^{j}_i}\}_{i \in \mathbb{N}}$ is a standard system. If this is true, we may apply the aforementioned extension theorem of Parthasarathy and also conclude that every probability measure on $\mathcal{F}_{T_\infty^{j}}$ has an extension to a probability measure on $\mathcal{F}_T.$

We have, from~Carr et al~(2014), that a sufficient condition for  $\{\mathcal{F}_{R^{j}_i}\}_{i \in \mathbb{N}}$ to be a standard system is the following: $\{\hat{\mathcal{F}}_t\}_{t \in [0,T]}:=\{\mathcal{F}_t \cap \mathcal{F}_{{T_\infty^{j}}}\}_{t \in [0,T]}$ is the right-continuous modification of a \textit{standard system}. In Carr et al~(2014), an example of an $\Omega$ and a filtration $\{\mathcal{F}_t\}_{t \in [0,T]}$  such that $\{\mathcal{F}_t \cap \mathcal{F}_{{T_\infty^{j}}}\}_{t \in [0,T]}$ is a right-continuous modification of a standard system is given:

 Let $E$ denote a locally compact space with a countable base (for example, $E= \mathbb{R}^{2d+2}$ for some $n \in \mathbb{N}$) and let $(\Omega)$ be the space of right-continuous paths $\omega : [0,T] \rightarrow [0, \infty] \times E$ whose $j^{th}$ component $\omega^{j}$ of $\omega$ is such that $\omega^{j}(T_\infty^{j}(\omega)+t)= \infty$ for all $t \geq 0$ and that have left limits on $(0, T_\infty^{j}(\omega))$ where $T_\infty^{j}(\omega)$ denotes the first time that $\omega^{j}= \infty.$ Let $\{\mathcal{F}^{0}_t\}_{t \in [0,T]}$ denote the filtration generated by the paths and $\{\mathcal{F}_t\}_{t \in [0,T]}$ its right-continuous modification. Then, it follows from works of any of Dellacherie~(1972), Meyer~(1972) and F\"ollmer~(1972; Example 6.3.2), that $\{\mathcal{F}_t \cap \mathcal{F}_{{T_\infty^{j}}}\}_{t \in [0,T]}$ is a right-continuous modification of a standard system.

In the example we are studying, we can equate the process $M^{j}_t$ with $\omega^{j},$ the $j^{th}$ component of $\omega.$ Thus, we have that, in our case, $\{\mathcal{F}_{{T^{j}_\infty}}\}_{i \in \mathbb{N}}$ is a standard system. 
 
Analogous to the argument used in Section 2 of~Carr et al~(2014), we also have that any probability measure $Q$ on $(\Omega, \mathcal{F}_{R^{j}})$ can be extended to a probability measure $\tilde{Q}$ on $(\Omega, \mathcal{F}_T).$

Note that we have, for all $A \in \mathcal{F}_{\nu \wedge T},$ and stopping times $\nu,$ 

\begin{multline*} P^{j}(A \cap \{R^{j} \textgreater \nu \wedge T\})=\lim_{i \to \infty}P^{j}(A\cap\{R^{j}_i \textgreater \nu \wedge T\})= \lim_{i \to \infty}P^{j}_i(A\cap\{R^{j}_i \textgreater \nu \wedge T\})= \\ \lim_{i \to \infty}E_{P}[1_{A}\cap\{R^{j}_i \textgreater \nu \wedge T\}(M^{j})^{R^{J}_i}]=  \lim_{i \to \infty}E_{P}[1_{A}\cap\{R^{j}_i \textgreater \nu \wedge T\} M^{j}_{\nu\wedge T}]=E_{P}[1_{A}{M^{j}}_{\nu \wedge T}] \end{multline*}

From this, we get~\eqref{1.3}. Taking $\nu=R^{j}_i$ and $A=\Omega$ we get $P^{j}(R^{j} \textgreater R^{j}_i \wedge T)=1$ $\forall i \in \mathbb{N}.$ \\

 Thus, we have that $P^{j}(A)=P^{j}(A \cap \{R \textgreater R^{j}_i \wedge T\})=E_P[(M^{j})^{R^{J}_i}_T1_{A}]$ for $A \in \mathcal{F}_{R^{j}_i}.$ Since $F_{{R^j}}=\bigvee_{i \in \mathbb{N}} \mathcal{F}_{R^{J}_i}$ and $\bigcup_{i \in \mathbb{N}}\mathcal{F}_{R_i}$ is a $\pi$ system, and by a standard application of the Monotone Class Theorem (cf, eg,~\cite[p. 7]{PP}) we have uniqueness of $P^{j}$ on $(\Omega, \mathcal{F}_{R^{j}}).$ 

We have that~\eqref{1.4} follows from~\eqref{1.3} from the monotone convergence theorem and~\eqref{1.5} follows from~\eqref{1.4} from $P^{j}(R^{j} \textgreater T)=1$ and applying~\eqref{1.3} to $U1_{\{R^{j} \textgreater \nu \wedge T\}}\tilde{M}^{j}_{\nu \wedge T}$ instead of $U$ for $U$ and $\nu$ as in Theorem~\ref{main}. 
 
 \end{proof}

We have the following important corollary to Theorem \ref{ext}: 

\bigskip

\begin{corollary}\label{col}
Components $1:k$ of $M$ are true martingales and components $k+1:d$ of $M$ are strict local martingales if and only if, for $j=1:k,$ $P^{j}(T^{j}_{\infty} \textgreater T)=1$ and for $j=k+1:d,$ $P^{j}(T^{j}_{\infty} \leq T) \textgreater 0.$  That is, if local martingales $M^{1}:M^{k}$ do not explode before time $T$ under measures $P^{1}:P^{k}$ and local martingales $M^{k+1}:M^{d}$ do explode before time $T$ under measures $P^{k+1}:P^{d}.$
\end{corollary}

\section{Explosions of Multidimensional Diffusions}\label{multi}

In this section, we will discuss and display some results found in~Stroock \& Varadhan (2006) that treat the subject of explosions of multidimensional diffusions. Unspecified citation in this section refer to that book. 

Let us assume that we have a probability space $(\Omega, \mathcal{F}_T, (\mathcal{F}_t)_{t\in [0,T]}, P).$

Let $X$ be an $\mathbb{R}^{d}$ valued multidimensional diffusion that solves

 \begin{equation} dX_t=\sigma(X_t)dB_t+b(X_t)dt; \quad X_0=x_0\label{d1}
 \end{equation}
Where in the above, $\sigma: \mathbb{R}^{d} \rightarrow \mathbb{R}^{d}$ is a $d\times d$ matrix of diffusion coefficients that is locally bounded, and $b: \mathbb{R}^{d} \rightarrow \mathbb{R}^{d}$ is a $d$ dimensional vector of locally bounded drift coefficients. $B$ is a $d$- dimensional $\mathbb{F}$ Brownian motion.

Let $\mathcal{L}$ be the extended generator of this diffusion $X$: \begin{equation} \mathcal{L}=\frac{d}{dt}+\frac{1}{2}\sum\limits_{i,j=1}^da_{i,j}(x)\frac{d^{2}}{dx_{i}dx_{j}}+\sum\limits_{i,j=1}^db^{i}(x)\frac{d}{dx_i} \end{equation}

In the above, $a=\sigma {\sigma}^{T}.$

We then have the following theorem regarding non-explosion:

\begin{theorem}\label{var1}[Theorem 10.2.1]\label{non_explosion}

Assume the existence of a non-negative function $V \in C^{1,2}([0,T] \times \mathbb{R}^{d})$ as well as the existence of a $\lambda \textgreater 0$ such that 
\begin{eqnarray}\label{lya} \lim_{|x| \to \infty} \inf_{0 \leq t \leq T}V(t,x)&=&\infty \\ 
\mathcal{L}V -\lambda V& \leq &0
\end{eqnarray} 

Then,  with probability $1,$ the process $X$ does not explode before $T.$ 
\end{theorem}

\begin{proof}

Define the sequence of stopping times $\tau_n=\inf\{t: |X_t| \geq n\}$. Since we have that $\mathcal{L} -\lambda V \leq 0,$ on $[0,T]  \times \mathbb{R}^{d},$ we obtain that \begin{equation*} V(0,X_0) \geq E[ \e^{-\lambda(T \wedge \tau_n)}V(T\wedge \tau_n, X_{T\wedge \tau_n})] \geq \e^{-\lambda T}E_{\tilde{P}}[V(\tau_n, X_{\tau_n}) 1_{\tau_n \textless T}]  \end{equation*} 

Since it is true that $|X_{\tau_n}|=n$ if $\tau_n \textless T$ and that we have assumed that $\lim_{|x| \to \infty} \inf_{0 \leq t \leq T}V(t,x)=\infty,$ we must have that \begin{equation*} \lim_{n \to \infty}P(\tau_n \textless T) =0 \end{equation*} Since the vector process $X$ does not explode before $T$, the process $S$ does not explode before $T,$ and we have $P(T_{\infty} \textgreater T) =1.$

\end{proof}

We also have the following theorem regarding explosion: 

\bigskip

\begin{theorem}\label{var2}[Also Theorem 10.2.1]

Assume the existence of a number $\lambda \textgreater 0$ and a bounded function $V \in C^{1,2}([0,T] \times \mathbb{R}^{d})$ such that \begin{equation}\label{expl} V(0,x_{0}) \textgreater \e^{-\lambda T} \sup_{x \in \mathbb{R}^{d}}V(T,x) \end{equation} and

\begin{equation*}\mathcal{L}V \geq \lambda V \end{equation*}

Then, we have $\lim_{n\to \infty} \tilde{P}(\tau_n \leq T) \textgreater 0.$

\end{theorem}

\begin{proof}
Define $\tau_n=\inf\{t: |X_t| \geq n\}.$ Now, we are supposing that $\mathcal{L}V \geq \lambda V$ for $t \in [0, \infty),$ we have:

 \begin{equation*} V(0,x_0)\leq  \e^{-\lambda T}(\sup_{x \in \mathbb{R}^{d}}V(T,x))P(\tau_n \textgreater T) + (\sup_{t \in [0,T]} \sup_{x \in \mathbb{R}^{d}}V(t,x))P(\tau_n \leq T) \end{equation*}

 If it were true that $\lim_{n \to \infty}P(\tau_n \leq T)$ were zero then we would arrive at \begin{equation*} V(0,x_0) \leq \e^{-\lambda T} \sup_{x \in \mathbb{R}^{2}}V(T,x) \end{equation*} which is a contradiction because we assumed that the function $V$ satisfied~\eqref{expl}. We are done, and we have established that we must have $P(\tau_\infty \leq T) \textgreater 0.$

\end{proof}

\begin{remark}\label{rem} Note that the conditions of this theorem are sufficient to ensure that \textit{each and every} component of the vector $X$ explodes. This can be seen by replacing the stopping times $\tau^{n}$ in the proof of the theorem by the stopping times $\tau^{i,n}=\inf\{t: |X^i_t| \geq n\}.$

\end{remark}

 We state two more theorems that involve conditions on the drift and diffusion coefficient of the vector diffusion $X$ such that explosion or non-explosion occurs:
 
 \bigskip
 
  \begin{theorem}~\label{var3}[Theorem 10.2.3]

Assume the existence of some $r$ and continuous functions $A: [r, \infty) \rightarrow (0,\infty)$ and $B:[r, \infty) \rightarrow (0,\infty)$ such that for $\rho \geq (2r)^{\frac{1}{2}},$  and $|x|= \rho,$

 \begin{eqnarray}\label{1023} A(\frac{\rho^{2}}{2})& \geq &\langle x, a(x)x \rangle \\&&   \langle x,a(x)x \rangle B(\frac{\rho^{2}}{2}) \geq (Trace(a(x)+2\langle x,b(x)\rangle \end{eqnarray} 
 
 and, with $C(\rho)=\e^{\int_{\frac{1}{2}}^{\rho}B(\sigma)d\sigma}:$ 
 
 $\int_{r}^{\infty}[C(\rho)]^{-1}d\rho \int_{\frac{1}{2}}^{\rho}\frac{C(\sigma)}{A(\sigma)}d\sigma= \infty.$

 Then, the process $X$ does not explode before time $T$.

 \end{theorem}
 \begin{remark} The conditions of this theorem ensure the existence of a function $V$ that satisfies the conditions of Theorem~\ref{var1}, which, as we know, guarantee non-explosion of the process $X.$
\end{remark}

 \begin{theorem}~\label{var4}[Theorem 10.2.4]
 
Assume that, for each $R \textgreater 0$ the following condition holds:\\ $\inf_{|\theta|=1}\inf_{|x|=R} \langle \theta, a(x)\theta \rangle \textgreater 0$ and $\sup_{|x| \leq R} |b(x)| \textless \infty.$ 
 
 Additionally, assume the existence of continuous functions $A: [\frac{1}{2}, \infty) \rightarrow (0,\infty)$ and $B:[\frac{1}{2}, \infty) \rightarrow (0,\infty)$ such that for $\rho \geq 1,$ and $|x|= \rho,$

 \begin{eqnarray}\label{1024} A(\frac{\rho^{2}}{2})& \leq &\langle x, a(x)x \rangle \\&&   \langle x,a(x)x \rangle B(\frac{\rho^{2}}{2}) \leq (Trace(a(x)+2\langle x,b(x)\rangle \end{eqnarray} 
 
 and, with $C(\rho)=\e^{\int_{\frac{1}{2}}^{\rho}B(\sigma)d\sigma}:$ 
 
 $\int_{\frac{1}{2}}^{\infty}[C(\rho)]^{-1}d\rho \int_{\frac{1}{2}}^{\rho}\frac{C(\sigma)}{A(\sigma)}d\sigma \textless \infty.$

Then, the process $X$ does explode before time $T.$
 
 \end{theorem}

\begin{remark}

The conditions of this theorem ensure that the following condition holds: 

\begin{equation} \lim_{T \to \infty} \lim_{n \to \infty} P^{n}[\tau^{n} \leq T]=1 \end{equation} Thus, there is a finite time before which the process $X$ explodes, with probability 1. 

 \end{remark}

\begin{proposition}
We can replace the assumption in Theorem~\ref{var3} that $\rho \geq (2r)^{\frac{1}{2}}$ and that $A: [r, \infty) \rightarrow (0,\infty)$ and $B:[r, \infty) \rightarrow (0,\infty)$ by the assumption that $\inf_{1 \leq i \leq d}\{x_i\} \geq 1$ and $A: [\frac{3}{2}, \infty) \rightarrow (0,\infty)$ and $B:[\frac{3}{2}, \infty) \rightarrow (0,\infty)$ and the same result holds.

We can also replace the assumption in Theorem~\ref{var4} that $\rho \geq 1$ and that $A: [\frac{1}{2}, \infty) \rightarrow (0,\infty)$ and $B:[\frac{1}{2}, \infty) \rightarrow (0,\infty)$ by the assumption that $\inf_{1 \leq i \leq d}\{x_i\} \geq 1$ and $A: [\frac{3}{2}, \infty) \rightarrow (0,\infty)$ and $B:[\frac{3}{2}, \infty) \rightarrow (0,\infty)$ and the same result holds.

We omit the proof of this proposition.

\end{proposition}

We are now ready to state our main theorem:

\begin{theorem}\label{main}

Let $x=\begin{bmatrix} x_1 \\x_2\\.\\.\\.\\x_{d+1}\end{bmatrix}.$
$x^{'}=\begin{bmatrix} x_1 \\x_2\\.\\.\\.\\x_{d}\end{bmatrix}.$
Let also $a=\sigma {\sigma}^{T}.$
Define $\beta^{j}_{i}(M_t,v_t)=(\Sigma)_{i,j}\sigma_{i,i}\sigma_{j,j}M^{j}.$
Let the vector diffusion process $X=\begin{bmatrix} M \\v \end{bmatrix}$ solve ~\eqref{d}. 

Assume that for $i=1:k$ there exist functions $A^{i}$ and $B^{i}$ satisfying the conditions in~\eqref{1023} such that

\begin{eqnarray}\label{1}
A^{i}(\frac{\rho^{2}}{2})& \geq &\langle x, a(x)x \rangle \\  \langle x,a(x)x \rangle B^{i}(\frac{\rho^{2}}{2})& \geq& (Trace(a(x)+2\langle x,b(x)\rangle + 2\langle x^{'},\beta^{i}(x)\rangle\\ \nonumber
&&+2\langle x_{d+1}, (\Sigma)_{i,d+1}\sigma_{i,i}(x)\bar{\sigma(x)}\rangle \end{eqnarray}

Assume also that for $i=k+1:d$ there exist functions $A^{i}$ and $B^{i}$ satisfying the conditions in~\eqref{1024} such that 

\begin{eqnarray}\label{2}
A^{i}(\frac{\rho^{2}}{2})& \leq &\langle x, a(x)x \rangle \\  
\langle x,a(x)x \rangle B^{i}(\frac{\rho^{2}}{2}) &\leq& (Trace(a(x)+2\langle x,b(x)\rangle + 2\langle x^{'},\beta^{i}(x)\rangle\\ \nonumber
&&+2\langle x_{d+1}, (\Sigma)_{i,d+1}\sigma_{i,i}(x)\bar{\sigma(x)}\end{eqnarray}

Then, components $1:k$ of $M$ are true martingales and components $k+1:d$ are strict local martingales.

\end{theorem}

\begin{proof}

 An application of Girsanov's theorem gives us that, under the measure $P^{j},$ the diffusion $X$ solves, up to an explosion time,

\begin{eqnarray} dM_t&=&M_t\sigma(M_t,v_t)dB_t+\beta^{j}(M_t,v_t); \quad M_0=1\label{d'}\\
 dv_t&=&\bar{\sigma}(M_t,v_t)dZ_t +b(M_t,v_t)dt+ (\Sigma)_{j,d+1}\sigma_{j,j}(M_t,v_t)\bar{\sigma}(M_t,v_t); \quad v_0=1\label{v'} \end{eqnarray}

Where the $i^{th}$ component of $\beta^{j}$ is given by $\beta^{j}_{i}=(\Sigma)_{i,j}\sigma_{i,i}\sigma_{j,j}M^{j}.$

Recall from Theorem~\ref{ext} that $M^{j}$ is a martingale if and only if $P(T^{j}_\infty \textgreater T)=1.$ That is, if and only if $M^{j}$ does not explode before time $T$ under the measure $P^{j}.$ Recall also from Corollary~\ref{col} that for components $1:k$ to be martingales and components $k+1:d$ to be strict local martingales, we need components $1:k$ not to explode under measures $P^{1}:P^{k}$ and for components $k+1:d$ to explode under measures $P^{k+1}:P^{d}.$

\bigskip

From Theorem~\ref{var3}, we have that the conditions in~\eqref{1} ensure that the vector diffusion $X$ does not explode before time $T$ under measures $P^{i}$ for $i=1:k$. From Theorem~\ref{var4}, we have that the conditions in~\eqref{2} ensure that the vector diffusion $X$ does explode before time $T$ for measures $P^{i}$ for $i=k+1:d.$ 

\bigskip

As we know from the Remark~\ref{rem}, this implies that the $M^{i}$ do not explode under measures $P^{i}$ for $i=1:k$ and that the $M^{i}$ do explode for the measures $P^{i}$ for $i=k+1:d.$ Thus, the conditions in Corollary~\ref{col} are satisfied, and this means that components $1:j$ of $M$ are true martingales and components $j+1:d$ are strict local martingales.

\iffalse The conditions in Theorem~\ref{main} are sufficient to ensure that the vector diffusion $X$ does not explode before time $T$ under measures $P^{i}$ for $i=1:k$ and that the vector diffusion $X$ does explode before time $T$ for measures $P^{i}$ for $i=k+1:d.$ 

\bigskip

As we know from Remark~\ref{rem}, this implies that the $M^{i}$ do not explode under measures $P^{i}$ for $i=1:k$ and that the $M^{i}$ do explode for the measures $P^{i}$ for $i=k+1:d.$ Thus, the conditions in corollary~\ref{col} are satisfied, and this means that components$1:j$ of $M$ are true martingales and components $j+1:d$ are strict local martingales. \fi

\end{proof}

\section{Examples}\label{examp}

Let us consider some examples. We will focus on the two and three-dimensional cases. 

\begin{example}

Suppose that we have a filtered probability space $(\Omega, \mathcal{F}_T, (\mathcal{F}_t)_{t\in [0,T]}, P)$ and one local martingale $S$ and a stochastic volatility $v$ that solve:

 \begin{eqnarray} dS_t&=&S_tf(S_t, v_t)dB_t; \quad S_0=1\\
 dv_t&=&\sigma(S_t,v_t)dW_t +b(S_t,v_t)dt;  \quad v_0=1
 \end{eqnarray}   
 
 Assume that $[B,W]=\rho.$ Suppose we want $S$ to be a martingale. For that to happen, we need it not to explode under the measure $P^{1}.$ Denoting by $\xi$ the explosion time of the process $X=\begin{bmatrix} S\\v \end{bmatrix},$ let us display the SDE satisfied by $X$ under $P^{1}:$

\begin{eqnarray*} dS_t&=&S_tf(S_t, v_t)dB_t1_{t\textless \xi}+ S_tf^{2}(S_t,v_t)1_{t\textless \xi}dt; \quad S_0=1\\
  dv_t&=&\sigma(S_t,v_t)1_{t\textless \xi}dW_t +b(S_t,v_t)1_{t\textless \xi}dt+ \rho \sigma(S_t,v_t)f(S_t,v_t)1_{t\textless \xi}dt;\quad v_0=1
   \end{eqnarray*}   

 We take the functions $A(x)$ and $B(x)$ from Theorem~\ref{var3} to be $A(x)=x^{1+\epsilon},$ $\epsilon \textless 1.$ and $B(x)=\frac{1}{x}.$ And we restrict these two functions to the interval $ [\frac{3}{2}, \infty).$ If we choose

 \begin{eqnarray*} f(x_1,x_2)=\frac{1}{2{x_1}^{2}}(\frac{|x|}{2})^{1+\epsilon} \\ \sigma(x_1,x_2)=\frac{1}{2x_2}(\frac{|x|}{2})^{1+\epsilon} \\ b(x_1,x_2)=-(\frac{1}{2{x_2}^{2}}+\frac{1}{x_1x_2})(\frac{|x|^2}{2})^{1+\epsilon}\\ \rho=-1 \end{eqnarray*} 
 
 Then it can be checked that $X$ does not explode under the measure $P^{1}$ and that $S$ is a true martingale. 
 
 Suppose we want $S$ to be a strict local martingale. We then take the functions $A(x)$ and $B(x)$ from Theorem~\ref{var4} to be $A(x)=x^{2+\epsilon},$ $\epsilon \textless 1$ and $B(x)=\frac{1}{x}.$  And we restrict these two functions to the interval $ [\frac{3}{2}, \infty).$
 
 If we choose \begin{eqnarray*} f(x_1,x_2)=\frac{1}{2{x_1}^{2}}(\frac{|x|}{2})^{2+\epsilon} \\ \sigma(x_1,x_2)=\frac{1}{2x_2}(\frac{|x|}{2})^{2+\epsilon} \\ b(x_1,x_2)=\frac{1}{x_2}|x|^{2+\epsilon} \\ \rho=1 \end{eqnarray*}

 Then the process $X$ explodes under the measure $P^{1},$ rendering $S$ a strict local martingale.
 
 \end{example}

\begin{example}

Suppose that we have a filtered probability space $(\Omega, \mathcal{F}_T, (\mathcal{F}_t)_{t\in [0,T]}, P)$ and two local martingales $S$ and $N$ and a stochastic volatility $v$ that solve

 \begin{eqnarray*} dS_t&=&S_tf(S_t, N_t,v_t)dB_t; \quad S_0=1\\
\label{example_eqn'}
dN_t&=&N_tg(S_t,N_t,v_t)dZ_t; \quad N_0=1\\
dv_t&=&\sigma(S_t,N_t,v_t)dW_t +b(S_t,N_t,v_t)dt ; \quad v_0=1
 \end{eqnarray*}   

 Let us write: \begin{eqnarray*} \big[B,W\big]=\rho^{1} \\ \big[B,Z\big]=\rho^{2} \\ \big[W,Z\big]= \rho^{3} \end{eqnarray*}

Let us display the SDE satisfied by $\begin{bmatrix} S\\ N \\v \end{bmatrix}$ under $P^{1}:$

 \begin{eqnarray*} dS_t&=&S_tf(S_t, N_t,v_t)dB_t1_{t \textless \xi}+ S_tf^{2}(S_t,N_t, v_t)1_{t \textless \xi}; \quad S_0=1\\
\label{x_eqn'}
dN_t&=&N_tg(S_t,X_t,v_t)dZ_t1_{t \textless \xi}+ \rho^{2}N_tg(S_t,X_t,v_t)\mu(S_t,N_t,v_t)1_{t \textless \xi}; \quad X_0=1\\
  dv_t&=&\sigma(S_t,N_t,v_t)dW_t1_{t \textless \xi} +b(S_t,N_t,v_t)dt1_{t \textless \xi} +\rho^{1}f(S_t,N_t,v_t)\sigma(S_t,N_t,v_t)dt1_{t \textless \xi} ; \quad v_0=1
 \end{eqnarray*}   
In the above $\xi$ is the explosion time of the vector process $X=\begin{bmatrix} S\\N\\v \end{bmatrix}.$

Let us display the SDE satisfied by $\begin{bmatrix} S\\ X \\v \end{bmatrix}$ under $P^{2}:$ 

 \begin{eqnarray*} dS_t&=&S_tf(S_t, N_t,v_t)dB_t1_{t \textless \xi}+ \rho^{2}S_tg(S_t,N_t,v_t)f(S_t,N_t, v_t)1_{t \textless \xi}; \quad S_0=1\\
\label{x_eqn''}
dX_t&=&X_tg(S_t,N_t,v_t)dZ_t1_{t \textless \xi}+ N_tg^{2}(S_t,N_t,v_t)1_{t \textless \xi}; \quad X_0=1\\
  dv_t&=&\sigma(S_t,X_t,v_t)dW_t1_{t \textless \xi} +b(S_t,X_t,v_t)dt1_{t \textless \xi}+ \rho^{3}g(S_t,X_t,v_t)\sigma(v_t)1_{t \textless \xi}dt;  \quad v_0=1
 \end{eqnarray*}

 Suppose we would like to display conditions such that $S$ is a martingale and $N$ is a strict local martingale. For this, we need $S$ to not explode under the measure $P^{1}$ and $N$ to explode under the measure $P^{2}.$

 We take the functions $A(x)$ and $B(x)$ from Theorem~\ref{var3} to be $A(x)=x^{1+\epsilon},$ $\epsilon \textless 1.$ and $B(x)=\frac{1}{x}.$ 
 
 We then take the functions $A(x)$ and $B(x)$ from Theorem~\ref{var4} to be $A(x)=x^{2+\epsilon},$ $\epsilon \textless 1$ and $B(x)=\frac{1}{x}.$  And we restrict all of these four functions to the interval $ [\frac{3}{2}, \infty).$
 
 Then it can be checked that if we choose $\rho^{1}=-1,$ $\rho^{2}=1,$ $\rho^{3}=1$ and
 \begin{eqnarray*} 
 f(x_1,x_2,x_3)=\frac{1}{\sqrt{3}}\frac{1}{{x_1}^{2}}(\frac{|x|}{2})^{1+\epsilon}\\
 g(x_1,x_2,x_3)=\frac{1}{\sqrt{3}}\frac{1}{{x_2}^{2}}(\frac{|x|}{2})^{1+\epsilon} \\
 \sigma(x_1,x_2,x_3)=\frac{1}{\sqrt{3}}\frac{1}{x_3}(\frac{|x|}{2})^{1+\epsilon}\\
 b(x_1,x_2,x_3)=(\frac{|x|}{2})^{2\epsilon}-(\frac{|x|}{2})^{2+2\epsilon}(\frac{5}{{x_1}^2}+\frac{1}{3{x_2}^2}+\frac{1}{3{x_3}^2})\end{eqnarray*}
 
 The conditions of Theorem~\ref{main} are satisfied, rendering $S$ a true martingale and $N$ a strict local martingale.

 \begin{remark}
In displaying sufficient conditions such that a local martingale is either a true martingale or a strict local martingale, we have exhibited conditions such that the entire vector $X,$ whose components are the $d$ local martingales and the stochastic volatility $v,$ does not explode or explodes under an appropriate different probability measure, respectively. Let's briefly discuss the case where we have just the price process and the stochastic volatility; that is, the case in which $X=\begin{bmatrix} S \\v \end{bmatrix}$ for a local martingale $S.$

Assume that $S$ and $v$ solve the following SDE:

 \begin{eqnarray} dS_t&=&S_tf(v_t)dB_t; \quad S_0=1\\
 dv_t&=&\sigma(v_t)dW_t +b(v_t)dt;  \quad v_0=1
 \end{eqnarray}

 In the case of $S$ being a martingale, recall that we have imposed the non-explosion of the vector process $\begin{bmatrix} S \\v \end{bmatrix}.$ The works of Mijatovic and Urusov in~\cite{MU} and that of Cui et. al. in~\cite{CUI}(specifically Theorem 4.1 in~\cite{CUI} which generalizes Theorem 2.1 in~\cite{MU}) imply that we are in one of the following cases:

 \begin{enumerate}
 \item $v$ does not hit either zero or $\infty$ under either of the measures $P$ or $P^{1}.$
 \item $v$ does not hit $\infty$ but hits zero under the measure $P^{1},$ and hits zero under the measure $P$ and doesn't hit $\infty$ under the measure $P.$ 
 \item $v$ does not hit $\infty$ but hits zero under the measure $P^{1},$ and hits zero and $\infty$ under the measure $P.$ 
  \item $v$ does not hit $\infty$ or zero under $P^{1}$ and hits $\infty$ but does not hit zero under $P.$
 
\item $v$ does not hit $\infty$ under $P$ or $P^{1}$, and $v$ hits zero under $P$ but doesn't hit zero under $P^{1}.$
 \end{enumerate}

 In the case of $S$ being a strict local martingale, according to the same authors, we are in one of the following cases: 
  \begin{enumerate}
  \item $v$ hits $\infty$ under $P^{1},$ does not under $P,$ $v$ hits zero under $P$ and $P^{1}.$   
   \item $v$ hits $\infty$ under $P^{1},$ does not under $P,$ $v$ hits zero under $P$ but not $P^{1}.$
   \item $v$ hits $\infty$ under $P^{1},$ does not under $P,$ $v$ does not hit zero under either $P$ or $P^{1}.$
   \item $v$ hits $\infty$ under $P^{1},$ does not under $P,$ $v$ does not hit zero under $P$ but hits zero under $P^{1}.$
   \item $v$ hits $0$ under $P^{1}$ but not under $P$ and hits $\infty$ under $P^{1}$ but not under $P.$ 
   \item $v$ hits $0$ under $P^{1}$ but not under $P$ and hits $\infty$ under both $P$ and $P^{1}.$ 
 \end{enumerate}

In the above, the measure $P^{1}$ is defined as in Theorem~\ref{ext}.

 \end{remark}

 \end{example}

\end{document}